\documentclass[copyright,creativecommons]{eptcs}
\providecommand{\event}{ICE 2021}

\usepackage{breakurl} 
\usepackage{underscore}

\usepackage{amsmath}
\usepackage{amsthm}
\usepackage{amssymb}
\usepackage{mathtools}
\usepackage{thm-restate}
\usepackage{stmaryrd}
\usepackage{bm}
\usepackage{tikz}
\usepackage{tikzit}

\tikzstyle{Process}=[fill=none, draw=none, shape=circle, inner sep=-.5]
\tikzstyle{block}=[fill=none, draw=black, shape=rectangle, rounded corners=.5mm, inner sep=.5mm]

\tikzstyle{crt}=[-, draw=cbBlue, line width=.8pt]
\tikzstyle{ci}=[-, draw=cbRed, line width=.8pt]
\tikzstyle{thick}=[-, line width=1pt]
\tikzstyle{thicker}=[-, line width=1.2pt]

\usetikzlibrary{positioning, shapes, calc}
\usepackage{wrapfig}
\usepackage{calc}
\usepackage{varwidth}
\usepackage{bussproofs}
\usepackage{xspace}
\usepackage[framemethod=tikz]{mdframed}
\usepackage{centernot}
\usepackage{xcolor}
\usepackage{cleveref}

\definecolor{cbBlue}{RGB}{14,40,142}
\definecolor{cbGreen}{RGB}{22,76,100}
\definecolor{cbRed}{RGB}{136,26,88}
\colorlet{lightgrey}{gray!20!white}
\definecolor{cbRef}{RGB}{165,54,6}
\colorlet{cbDkRef}{cbRef!80!white}
\colorlet{cbDkGreen}{cbGreen!90!white}
\hypersetup{
    hypertexnames=false,
    colorlinks=true,
    final,
    citecolor=cbGreen,
    linkcolor=cbRef
}

\newcommand{\etal}{\emph{et al.}\xspace}

\newcommand{\sff}[1]{\relax\ifmmode\mathsf{#1}\else\textsf{#1}\fi}
\newcommand{\mathsc}[1]{\text{\normalfont\scshape#1}}
\newcommand{\scc}[1]{\relax\ifmmode\mathsc{#1}\else\textsc{#1}\fi}
\newcommand{\mbb}[1]{\mathbb{#1}}
\newcommand{\<}{\langle}
\renewcommand{\>}{\rangle}
\newcommand{\sepr}{\;\mbox{\large{$\mid$}}\;}
\newcommand{\redd}{\mathbin{\longrightarrow}}
\newcommand{\nredd}{{\centernot\longrightarrow}}
\newcommand{\rred}[1]{\shortrightarrow_{#1}}
\newcommand{\brred}[1]{\beta_{#1}}
\newcommand{\krred}[1]{\kappa_{#1}}
\newcommand{\fn}{\mathrm{fn}}
\newcommand{\frv}{\mathrm{frv}}
\newcommand{\an}{\mathrm{an}}
\newcommand{\live}{\mathrm{live}}
\newcommand{\ol}[1]{\overline{#1}}
\newcommand{\rott}[1]{\mathpalette\rot{#1}}
\newcommand{\rot}[2]{\rotatebox[origin=c]{180}{$#1{#2}$}}
\newcommand{\pr}{\mathrm{pr}}

\newcommand{\figref}[1]{Fig.~\labelcref{#1}}
\newcommand{\mirr}[1]{\mathpalette\mir{#1}}
\newcommand{\mir}[2]{\reflectbox{$#1{#2}$}}
\newcommand{\iddots}{\mirr{\ddots}}

\let\oprod\prod
\renewcommand{\prod}{\mathchoice{\textstyle}{}{}{}{\oprod}}

\newcommand{\puts}{\mathbin{\triangleleft}}
\renewcommand{\gets}{\mathbin{\triangleright}}
\newcommand{\call}[1]{\<#1\>}
\let\onu\nu
\renewcommand{\nu}[1]{(\mkern-1mu\bm{\onu} #1)}
\let\oprl\|
\renewcommand{\|}{\mathbin{|}}
\newcommand{\0}{\bm{0}}
\newcommand{\fwd}{\mathbin{\leftrightarrow}}
\newcommand{\subst}[1]{{\mathchoice{\scriptstyle}{\scriptstyle}{}{}{\{#1\}}}}
\newcommand{\rsubst}[1]{\{#1\}}

\newcommand{\tensor}{\ensuremath{\mathbin{\otimes}}}
\newcommand{\parr}{\mathbin{\rott{\&}}}
\newcommand{\pri}{\mathsf{o}}
\newcommand{\opri}{\kappa}

\newcommand{\1}{\bm{1}}
\newcommand{\lift}[1]{{\uparrow\mkern1mu}^{#1}}

\newtheorem{theorem}{Theorem}
\newtheorem{lemma}[theorem]{Lemma}
\newtheorem{definition}{Definition}

\newtheorem{notation}{Notation}

\newcommand{\infAx}[2]{\AxiomC{} \RightLabel{#2} \UnaryInfC{#1}}
\newcommand{\infAss}[1]{\AxiomC{#1}}

\newcommand{\infUn}[2]{\RightLabel{#2} \UnaryInfC{#1}}
\newcommand{\infDblUn}[2]{\RightLabel{#2} \doubleLine\UnaryInfC{#1}}
\newcommand{\infBin}[2]{\RightLabel{#2} \BinaryInfC{#1}}
\newcommand{\infDblBin}[2]{\RightLabel{#2} \doubleLine\BinaryInfC{#1}}

\newenvironment{wfit}{\begin{varwidth}{\textwidth}}{\end{varwidth}}

\title{Deadlock Freedom for~ \\ Asynchronous and Cyclic Process Networks%
\thanks{Research partially supported by the Dutch Research Council (NWO) under project No. 016.Vidi.189.046 (Unifying Correctness for Communicating Software).}}

\makeatletter
\let\thetitle\@title
\makeatother

\def\titlerunning{Deadlock Freeodm for Asynchronous and Cyclic Process Networks}

\author{
    Bas van den Heuvel and Jorge A.\ P\'erez
    \institute{University of Groningen, The Netherlands}
}
\def\authorrunning{Van den Heuvel \& P\'erez}

\begin{document}

\maketitle

\begin{abstract}
    This paper considers the challenging problem of establishing deadlock freedom for message-passing processes using behavioral type systems.
    In particular, we consider the case of processes that implement session types by communicating asynchronously in cyclic process networks.
    We present APCP, a typed process framework for deadlock freedom which supports asynchronous communication, delegation, recursion, and a general form of process composition that enables specifying cyclic process networks.
    We discuss the main decisions involved in the design of APCP and illustrate its expressiveness and flexibility using several examples.
\end{abstract}

\section{Introduction}
\label{s:intro}

Modern software systems often comprise independent components that interact by passing messages.
The $\pi$-calculus is a consolidated formalism for specifying and reasoning about message-passing processes~\cite{book/Milner89,journal/ic/MilnerPW92}.
Type systems for the $\pi$-calculus can statically enforce communication correctness.
In this context, \emph{session types} are a well-known approach, describing two-party communication protocols for channel endpoints and enforcing properties such as \emph{protocol fidelity} and \emph{deadlock freedom}.

Session type research has gained a considerable impulse after the discovery by Caires and Pfenning~\cite{conf/concur/CairesP10} and Wadler~\cite{conf/icfp/Wadler12} of Curry-Howard correspondences between session types and linear logic~\cite{journal/tcs/Girard87}.
Processes typable in type systems derived from these correspondences are inherently deadlock free.
This is because the \scc{Cut}-rule of linear logic imposes that processes in parallel must connect on exactly one pair of dual endpoints.
However, whole classes of deadlock free processes are not expressible with the restricted parallel composition and endpoint connection resulting from \scc{Cut}~\cite{conf/express/DardhaP15}.
Such classes comprise \emph{cyclic process networks} in which parallel components are connected on multiple endpoints at once.
Defining a type system for deadlock free, cyclic processes is challenging, because such processes may contain \emph{cyclic dependencies}, where components are stuck waiting for each other.

Advanced type systems that enforce deadlock freedom of cyclic process networks are due to Ko\-ba\-ya\-shi~\cite{conf/concur/Kobayashi06}, who exploits \emph{priority annotations} on types to avoid circular dependencies.
Dardha and Gay bring these insights to the realm of session type systems based on linear logic by defining Priority-based CP (PCP)~\cite{conf/fossacs/DardhaG18}.
Indeed, PCP incorporates the type annotations of Padovani's simplification of Kobayashi's type system~\cite{conf/csl/Padovani14} into Wadler's Classical Processes (CP) derived from classical linear logic~\cite{conf/icfp/Wadler12}.

In this paper, we study the effects of \emph{asynchronous communication} on type systems for deadlock free cyclic process networks.
To this end, we define \emph{Asynchronous PCP} (APCP), which combines Dardha and Gay's type annotations with DeYoung \etal's semantics for asynchronous communication~\cite{conf/csl/DeYoungCPT12}, and adds support for tail recursion.
APCP uncovers fundamental properties of type systems for asynchronous communication, and simplifies PCP's type annotations while preserving deadlock freedom results.

In \Cref{s:milner}, we motivate APCP by discussing Milner's cyclic scheduler~\cite{book/Milner89}.
\Cref{s:apcp} defines APCP's language and type system, and proves Type Preservation (\Cref{t:subjectRed}) and Deadlock Freedom (\Cref{t:closedDLFree}).
In \Cref{s:examples}, we showcase APCP by returning to Milner's cyclic scheduler and using examples inspired by Padovani~\cite{conf/csl/Padovani14} to illustrate asynchronous communication and deadlock detection. 
\Cref{s:concl} discusses related work and draws conclusions.

\section{Motivating Example: Milner's Cyclic Scheduler}
\label{s:milner}

We motivate by example the development of APCP, our type system for deadlock freedom in asynchronous, cyclic message-passing processes.
We consider Milner's cyclic scheduler~\cite{book/Milner89}, which crucially relies on asynchrony and recursion.
This example is inspired by Dardha and Gay~\cite{conf/fossacs/DardhaG18}, who use PCP to type a synchronous, non-recursive version of the scheduler.

\newsavebox{\boxCyclic}
\newlength{\wCyclic}
\savebox{\boxCyclic}{\begin{tikzpicture}
	\begin{pgfonlayer}{nodelayer}
		\node [style=Process] (0) at (-4, 4) {};
		\node [style=Process] (1) at (-4, 4) {$A_1$};
		\node [style=Process] (2) at (-2, 5) {$A_2$};
		\node [style=Process] (3) at (0, 4) {$A_3$};
		\node [style=Process] (4) at (0, 2.75) {$A_4$};
		\node [style=Process] (5) at (-2, 1.75) {$A_5$};
		\node [style=Process] (6) at (-4, 2.75) {$A_6$};
		\node [style=Process] (7) at (-5, 5) {$P_1$};
		\node [style=Process] (8) at (-2, 6.25) {$P_2$};
		\node [style=Process] (9) at (1, 5) {$P_3$};
		\node [style=Process] (10) at (1, 1.75) {$P_4$};
		\node [style=Process] (11) at (-2, 0.5) {$P_5$};
		\node [style=Process] (12) at (-5, 1.75) {$P_6$};
	\end{pgfonlayer}
	\begin{pgfonlayer}{edgelayer}
		\draw [style=thick, bend left, looseness=0.75] (1) to (2);
		\draw [style=thick, bend left, looseness=0.75] (2) to (3);
		\draw [style=thick, bend left=15, looseness=0.75] (3) to (4);
		\draw [style=thick, bend left, looseness=0.75] (4) to (5);
		\draw [style=thick, bend right=330, looseness=0.75] (5) to (6);
		\draw [style=thick, bend left=15, looseness=0.75] (6) to (1);
		\draw [style=thick] (7) to (1);
		\draw [style=thick] (3) to (9);
		\draw [style=thick] (4) to (10);
		\draw [style=thick] (6) to (12);
		\draw [style=thicker] (8) to (2);
		\draw [style=thicker] (5) to (11);
	\end{pgfonlayer}
\end{tikzpicture}
}
\settowidth{\wCyclic}{\usebox{\boxCyclic}}

\begin{wrapfigure}{R}{\wCyclic+12mm}
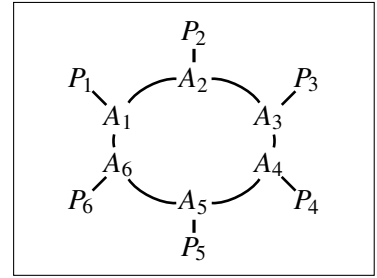

    \vspace{-1.5em}
    \begin{mdframed}
        \ctikzfig{cyclic}
        \vspace{-.5em}
    \end{mdframed}
    \vspace{-0.5em}
    \caption{Milner's cyclic sche\-duler with 6 workers. $\protect\phantom{xxxxxxx}$ Lines denote channels connecting processes.}
    \label{f:milner}
\end{wrapfigure}

The system consists of $n \geq 1$ worker processes $P_i$ (the workers, for short), each attached to a partial scheduler $A_i$.
The partial schedulers connect to each other in a ring structure, together forming the cyclic scheduler.
The scheduler then lets the workers perform their tasks in rounds, each new round triggered by the leading partial scheduler $A_1$ (the \emph{leader}) once each worker finishes their previous task.
We refer to the non-leading partial schedulers $A_{i+1}$ for $1 \leq i < n$ as the \emph{followers}.

Each partial scheduler $A_i$ has a channel endpoint $a_i$ to connect with the worker $P_i$'s channel endpoint $b_i$.
The leader $A_1$ has an endpoint $c_n$ to connect with $A_n$ and an endpoint $d_1$ to connect with $A_2$ (or with $A_1$ if $n=1$; we further elude this case for brevity).
Each follower $A_{i+1}$ has an endpoint $c_i$ to connect with $A_i$ and an endpoint $d_{i+1}$ to connect with $A_{i+2}$ (or with $A_1$ if $i+1=n$; we also elude this case).

In each round of the scheduler, each follower $A_{i+1}$ awaits a start signal from $A_i$, and then asynchronously signals $P_{i+1}$ and $A_{i+2}$ to start.
After awaiting acknowledgment from $P_{i+1}$ and a next round signal from $A_i$, the follower then signals next round to $A_{i+2}$.
The leader $A_1$, responsible for starting each round of tasks, signals $A_2$ and $P_1$ to start, and, after awaiting acknowledgment from $P_1$, signals next round to $A_2$.
Then, the leader awaits $A_n$'s start and next round signals.

It is crucial that $A_1$ does not await $A_n$'s start signal before starting $P_1$, as the leader would otherwise not be able to initiate rounds of tasks.
Asynchrony thus plays a central role here: because $A_n$'s start signal is non-blocking, it can start $P_n$ before $A_1$ has received the start signal.
Of course, $A_1$ does not need to await $A_n$'s start and next round signals to make sure that every partial scheduler is ready to start the next round.

Let us specify the partial schedulers formally:
\begin{align*}
    A_1 &:= \mu X(a_1,c_n,d_1); d_1 \puts \sff{start} \cdot a_1 \puts \sff{start} \cdot a_1 \gets \sff{ack}; d_1 \puts \sff{next} \cdot c_n \gets \sff{start}; c_n \gets \sff{next}; X\call{a_1,c_n,d_1}
    \\
    A_{i+1} &:= \mu X(a_{i+1},c_i,d_{i+1}); \begin{array}[t]{l}
        c_i \gets \sff{start}; a_{i+1} \puts \sff{start} \cdot d_{i+1} \puts \sff{start} \cdot a_{i+1} \gets \sff{ack};
        \\
        c_i \gets \sff{next}; d_{i+1} \puts \sff{next} \cdot X\call{a_{i+1},c_i,d_{i+1}}
    \end{array}
    \qquad \forall 1 \leq i < n
\end{align*}
The syntax `$\mu X(\tilde{x}); P$' denotes a recursive loop where $P$ has access to the endpoints in $\tilde{x}$ and $P$ may contain recursive calls `$X\call{\tilde{y}}$' where the endpoints in $\tilde{y}$ are assigned to $\tilde{x}$ in the next round of the loop.
The syntax `$x \puts \ell$' denotes the output of label $\ell$ on $x$, and `$x \gets \ell$' denotes the input of label $\ell$ on $x$.
Outputs are non-blocking, denoted `$\cdot$', whereas inputs are blocking, denoted `$;$'.
For example, process $x \puts \ell \cdot y \gets \ell'; P$ may receive $\ell'$ on $y$ and continue as $x \puts \ell \cdot P$.

Leaving the workers unspecified, we formally specify the complete scheduler as a ring of partial schedulers connected to workers:
\[
    Sched_n := \nu{c_i d_i}_{1 \leq i \leq n}(\prod_{1 \leq i \leq n} \nu{a_i b_i}(A_i \| P_i))
\]
The syntax `$\nu{x y}P$' denotes the connection of endpoints $x$ and $y$ in $P$, and `$\prod_{i \in I} P_i$' and `$P \| Q$' denote parallel composition.
\Cref{f:milner} illustrates $Sched_6$, the scheduler for six workers.

We return to this example in \Cref{s:examples}, where we type check the scheduler using APCP to show that it is deadlock free.

\section{APCP: Asynchronous Priority-based Classical Processes}
\label{s:apcp}

In this section, we define APCP, a linear type system for $\pi$-calculus processes that communicate asynchronously (i.e., the output of messages is non-blocking) on connected channel endpoints.
In APCP, processes may be recursive and cyclically connected.
Our type system assigns to endpoints types that specify two-party protocols, in the style of binary session types~\cite{conf/concur/Honda93}.

APCP combines the salient features of Dardha and Gay's Priority-based Classical Processes (PCP)~\cite{conf/fossacs/DardhaG18} with DeYoung \etal's semantics for asynchronous communication~\cite{conf/csl/DeYoungCPT12}, both works inspired by Curry-Howard correspondences between linear logic and session types~\cite{conf/concur/CairesP10,conf/icfp/Wadler12}.
Recursion---not present in the works by Dardha and Gay and DeYoung \etal---is an orthogonal feature, whose syntax is inspired by the work of Toninho \etal~\cite{conf/tgc/ToninhoCP14}.

As in PCP, types in APCP rely on \emph{priority} annotations, which enable cyclic connections by ruling out circular dependencies between sessions.
A key insight of our work is that asynchrony induces significant improvements in priority management: the non-blocking outputs of APCP do not need priority checks, whereas PCP's outputs are blocking and thus require priority checks.

Properties of well-typed APCP processes are \emph{type preservation} (\Cref{t:subjectRed}) and \emph{deadlock freedom} (\Cref{t:closedDLFree}).
This includes cyclically connected processes, which priority-annotated types guarantee free from circular dependencies that may cause deadlock.

\subsection{The Process Language}
\label{ss:procs}

\begin{figure}[t]
    \begin{mdframed}
        Process syntax:
        \begin{align*}
            P,Q ::=
            &~ x[y,z]
            & \text{output}
            & ~~~ \sepr ~~
              x(y,z); P
            & \text{input}
            \\[-3pt]
            \sepr\!
            &~ x[z] \triangleleft i
            & \text{selection}
            & ~~~ \sepr ~~
              x(z) \triangleright \{i: P_i\}_{i \in I}
            & \text{branching}
            & ~~~ \sepr ~~
              \nu{x y}P
            & \text{restriction}
            \\[-3pt]
            \sepr\!
            &~ (P \| Q)
            & \text{parallel}
            & ~~~ \sepr ~~
              \0
            & \text{inaction}
            & ~~~ \sepr ~~
              x \fwd y
            & \text{forwarder}
            \\[-3pt]
            \sepr\!
            &~ \mu X(\tilde{x}); P
            & \text{recursive loop}
            & ~~~ \sepr ~~
              X\call{\tilde{x}}
            & \text{recursive call}
        \end{align*}

        \vspace{-1em}
        \hbox to \textwidth{\leaders\hbox to 3pt{\hss . \hss}\hfil}

        \smallskip
        Structural congruence:
        \begin{align*}
            P \equiv_\alpha P' \implies {}
            &
            &
            P
            &\equiv P'
            &
            x \fwd y
            &\equiv y \fwd x
            \\
            &
            &
            P \| Q
            &\equiv Q \| P
            &
            \nu{x y} x \fwd y
            &\equiv \0
            \\
            &
            &
            P \| \0
            &\equiv P
            &
            P \| (Q \| R)
            &\equiv (P \| Q) \| R
            \\
            x,y \notin \fn(P) \implies {}
            &
            &
            P \| \nu{x y}Q
            &\equiv \nu{x y}(P \| Q)
            &
            \nu{x y}\0
            &\equiv \0
            \\
            |\tilde{x}| = |\tilde{y}| \implies {}
            &
            &
            \mu X(\tilde{x}); P
            &\equiv P \rsubst{\mu X(\tilde{y}); P \subst{\tilde{y}/\tilde{x}} / X\call{\tilde{y}}}
            &
            \nu{x y}P
            &\equiv \nu{y x}P
            \\
            &
            &
            &
            &
            \nu{x y}\nu{z w} P
            &\equiv \nu{z w}\nu{x y} P
        \end{align*}

        \vspace{-1em}
        \hbox to \textwidth{\leaders\hbox to 3pt{\hss . \hss}\hfil}

        \smallskip
        Reduction:
        \begin{align*}
            & \brred{\scc{Id}}
            &
            z,y \neq x \implies {}
            &
            &
            \nu{y z}( x \fwd y \| P )
            &\redd P \subst{x/z}
            \\
            & \brred{\tensor\parr}
            &
            &
            &
            \nu{x y}( x[a,b] \| y(v,z); P )
            & \redd P \subst{a/v, b/z}
            \\
            & \brred{\oplus\&}
            &
            j \in I \implies {}
            &
            &
            \nu{x y}( x[b] \triangleleft j \| y(z) \triangleright \{i: P_i\}_{i \in I} )
            &\redd P_j \subst{b/z}
            \\[8pt]
            & \krred{\parr}
            &
            x \notin \tilde{v}, \tilde{w} \implies {}
            &
            &
            \nu{\tilde{v} \tilde{w}}(x(y,z); P \| Q)
            &\redd x(y,z); \nu{\tilde{v} \tilde{w}}(P \| Q)
            \\
            & \krred{\&}
            &
            x \notin \tilde{v}, \tilde{w} \implies {}
            &
            & \nu{\tilde{v} \tilde{w}}(x(z) \triangleright \{i: P_i\}_{i \in I} \| Q)
            & \redd x(z) \triangleright \{i: \nu{\tilde{v} \tilde{w}}(P_i \| Q)\}_{i \in I}
        \end{align*}

        \noindent
        \mbox{}\hfill%
        \begin{wfit}
            \begin{prooftree}
                \infAss{
                    $(P \equiv P') \wedge (P' \redd Q') \wedge (Q' \equiv Q)$
                }
                \infUn{
                    $P \redd Q$
                }{$\rred{\equiv}$}
            \end{prooftree}
        \end{wfit}%
        \hfill%
        \begin{wfit}
            \begin{prooftree}
                \infAss{
                    $\raisebox{9pt}{} P \redd Q$
                }
                \infUn{
                    $\nu{x y} P \redd \nu{x y} Q$
                }{$\rred{\onu}$}
            \end{prooftree}
        \end{wfit}%
        \hfill%
        \begin{wfit}
            \begin{prooftree}
                \infAss{
                    $\raisebox{9pt}{} P \redd Q$
                }
                \infUn{
                    $P \| R \redd Q \| R$
                }{$\rred{\|}$}
            \end{prooftree}
        \end{wfit}%
        \hfill\mbox{}
    \end{mdframed}

    \caption{Definition of APCP's process language.}
    \label{f:procdef}
\end{figure}

We consider an asynchronous $\pi$-calculus~\cite{conf/ecoop/HondaT91,report/Boudol92}.
We write $x, y, z, \ldots$ to denote (channel) \emph{endpoints} (also known as \emph{names}), and write $\tilde{x}, \tilde{y}, \tilde{z}, \ldots$ to denote sequences of endpoints.
Also, we write $i, j, k, \ldots$ to denote \emph{labels} for choices and $I, J, K, \ldots$ to denote sets of labels.
We write $X, Y, \ldots$ to denote \emph{recursion variables}, and $P,Q, \ldots$ to denote processes.

\Cref{f:procdef} (top) gives the syntax of processes.
The output action `$x[y,z]$' sends a message $y$ (an endpoint) and a continuation endpoint $z$ along $x$.
The input prefix `$x(y,z); P$' blocks until a message and a continuation endpoint are received on $x$ (referred to in $P$ as the placeholders $y$ and $z$, respectively), binding $y$ and $z$ in $P$.
The selection action `$x[z] \puts i$' sends a label $i$ and a continuation endpoint $z$ along $x$.
The branching prefix `$x(z) \gets \{i: P_i\}_{i \in I}$' blocks until it receives a label $i \in I$ and a continuation endpoint (reffered to in $P_i$ as the placeholder $z$) on $x$, binding $z$ in each $P_i$.
Restriction `$\nu{x y} P$' binds $x$ and $y$ in $P$, thus declaring them as the two endpoints of the same channel and enabling communication, as in Vasconcelos~\cite{journal/ic/Vasconcelos12}.
The process `$(P \| Q)$' denotes the parallel composition of $P$ and $Q$.
The process `$\0$' denotes inaction.
The forwarder process `$x \fwd y$' is a primitive copycat process that links together $x$ and $y$.
The prefix `$\mu X(\tilde{x}); P$' defines a recursive loop, binding occurrences of $X$ in $P$; the endpoints $\tilde{x}$ form a context for $P$.
The recursive call `$X\call{\tilde{x}}$' loops to its corresponding $\mu X$, providing the endpoints $\tilde{x}$ as context.
We only consider contractive recursion, disallowing processes with subexpressions of the form `$\mu X_1(\tilde{x}); \ldots; \mu X_n(\tilde{x}); X_1\call{\tilde{x}}$'.

Endpoints and recursion variables are free unless otherwise stated (i.e., unless they are bound somehow).
We write `$\fn(P)$' and `$\frv(P)$' for the sets of free names and free recursion variables of $P$, respectively.
Also, we write `$P \subst{x/y}$' to denote the capture-avoiding substitution of the free occurrences of $y$ in $P$ for $x$.
The notation `$P\rsubst{\mu X(\tilde{y}); P' / X\call{\tilde{y}}}$' denotes the substitution of occurrences of recursive calls `$X\call{\tilde{y}}$' in $P$ with the recursive loop `$\mu X(\tilde{y}); P'$', which we call \emph{unfolding} recursion.
We write sequences of substitutions `$P \subst{x_1/y_1} \ldots \subst{x_n/y_n}$' as `$P \subst{x_1/y_1, \ldots, x_n/y_n}$'.

Except for asynchrony and recursion, there are minor differences with respect to the languages of Dardha and Gay~\cite{conf/fossacs/DardhaG18} and DeYoung \etal~\cite{conf/csl/DeYoungCPT12}.
Unlike Dardha and Gay's, our syntax does not include empty input and output prefixes that explicitly close channels; this simplifies the type system.
We also do not include the operator for replicated servers, denoted `${!}x(y); P$', which is present in the works by both Dardha and Gay and DeYoung \etal
Although replication can be handled without difficulties, we omit it here; we prefer focusing on recursion, because it fits well with the examples we consider.
We discuss further these omitted constructs in \Cref{ss:exrules}.

\paragraph{Simplifying Notation}

In an output `$x[y,z]$', both $y$ and $z$ are free; they can be bound to a continuation process using parallel composition and restriction, as in $\nu{y a}\nu{z b}(x[y,z] \| P_{a,b})$.
The same applies to selection `$x[z] \puts i$'.
We introduce useful notations that elide the restrictions and continuation endpoints:

\begin{notation}[Derivable Actions and Prefixes]\label{n:sugar}
    We use the following syntactic sugar:
    \begin{align*}
        \ol{x}[y] \cdot P
        & := \nu{y a}\nu{z b}(x[a,b] \| P\subst{z/x})
        &
        \ol{x} \puts \ell \cdot P
        & := \nu{z b}(x[b] \puts \ell \| P\subst{z/x})
        \\
        x(y); P
        & := x(y,z); P\subst{z/x}
        &
        x \gets \{i:P_i\}_{i \in I}
        & := x(z) \gets \{i:P_i\subst{z/x}\}_{i \in I}
    \end{align*}
\end{notation}

\noindent
Note the use of `$\,\cdot\,$' instead of~`$\,;\,$' in output and selection to stress that they are non-blocking.

\paragraph{Operational Semantics}

We define a reduction relation for processes ($P \redd Q$) that formalizes how complementary actions on connected endpoints may synchronize.
As usual for $\pi$-calculi, reduction relies on \emph{structural congruence} ($P \equiv Q$), which  equates the behavior of processes with minor syntactic differences; it is the smallest congruence relation satisfying the axioms in \Cref{f:procdef} (middle).

Structural congruence defines the following properties of our process language.
Processes are equivalent up to $\alpha$-equivalence.
Parallel composition is associative and commutative, with unit `$\0$'.
The forwarder process is symmetric, and equivalent to inaction if both endpoints are bound together through restriction.
A parallel process may be moved into or out of a restriction as long as the bound channels do not appear free in the moved process: this is \emph{scope inclusion} and \emph{scope extrusion}, respectively.
Restrictions on inactive processes may be dropped, and the order of endpoints in restrictions and of consecutive restrictions does not matter.
Finally, a recursive loop is equivalent to its unfolding, replacing any recursive calls with copies of the recursive loop, where the call's endpoints are pairwise substituted for the contextual endpoints of the loop (this is \emph{equi-recursion}; see, e.g., Pierce~\cite{book/Pierce02}).

We can now define our reduction relation.
Besides synchronizations, reduction includes \emph{commuting conversions}, which allow pulling prefixes on free channels out of restrictions; they are not necessary for deadlock freedom, but they are usually presented in Curry-Howard interpretations of linear logic~\cite{conf/concur/CairesP10,conf/icfp/Wadler12,conf/fossacs/DardhaG18,conf/csl/DeYoungCPT12}.
We define the reduction relation `$P \redd Q\mkern2mu$' by the axioms and closure rules in \Cref{f:procdef} (bottom).
Axioms labeled `$\beta$' are \emph{synchronizations} and those labeled `$\kappa$' are  {commuting conversions}.
We write `$\redd_{\beta}$' for reductions derived from $\beta$-axioms, and `$\redd^\ast$' for the reflexive, transitive closure of `$\redd$'.

Rule $\beta_{\scc{Id}}$ implements the forwarder as a substitution.
Rule $\beta_{\tensor \parr}$ synchronizes an output and an input on connected endpoints and substitutes the message and continuation endpoint.
Rule $\beta_{\oplus \&}$ synchronizes a selection and a branch:
the received label determines the continuation process, substituting the continuation endpoint appropriately.
Rule $\kappa_{\parr}$ (resp.\ $\kappa_{\&}$) pulls an input (resp.\ a branching) prefix on free channels out of enclosing restrictions.
Rules $\rightarrow_\equiv$, $\rightarrow_\onu$, and $\rightarrow_{\|}$ close reduction under structural congruence, restriction, and parallel composition, respectively.

Notice how output and selection actions send free names.
This is different from the works by Dardha and Gay~\cite{conf/fossacs/DardhaG18} and DeYoung \etal~\cite{conf/csl/DeYoungCPT12}, where, following an internal mobility discipline~\cite{journal/tcs/Boreale98}, communication involves bound names only.
As we show in the next subsection, this kind of \emph{bound output} is derivable (cf.\ \Cref{t:admissible}).

\subsection{The Type System}
\label{ss:typesys}

APCP types processes by assigning binary session types to channel endpoints.
Following Curry-Howard interpretations, we present session types as linear logic propositions (cf., e.g., Wadler~\cite{conf/icfp/Wadler12}, Caires and Pfenning~\cite{conf/esop/CairesP17}, and Dardha and Gay~\cite{conf/fossacs/DardhaG18}).
We extend these propositions with recursion and \emph{priority} annotations on connectives.
Intuitively, actions typed with lower priority should be performed before those with higher priority.
We write $\pri, \opri, \pi, \rho, \ldots$ to denote priorities, and `$\omega$' to denote the ultimate priority that is greater than all other priorities  and cannot be increased further.
That is, $\forall t \in \mbb{N}.~\omega > t$ and $\forall t \in \mbb{N}.~\omega + t = \omega$.

\emph{Duality}, the cornerstone of session types and linear logic, ensures that the two endpoints of a channel have matching actions.
Furthermore, dual types must have matching priority annotations.
The following inductive definition of duality suffices for our tail-recursive types (cf.\ Gay \etal~\cite{conf/places/GayTV20}).

\begin{definition}[Session Types and Duality]\label{d:props}
    The following grammar defines the syntax of \emph{session types} $A,B$, followed by the dual $\ol{A},\ol{B}$ of each type.
    Let $\pri \in \mbb{N} \cup \{\omega\}$.
    \begin{align*}
        A,B
        &::= A \tensor^\pri B \sepr A \parr^\pri B \sepr {\oplus}^\pri \{i: A_i\}_{i \in I} \sepr \&^\pri \{i: A_i\}_{i \in I} \sepr \bullet \sepr \mu X.A \sepr X
        \\
        \ol{A},\ol{B}
        &::= \ol{A} \parr^\pri \ol{B} \sepr \ol{A} \tensor^\pri \ol{B} \sepr \&^\pri \{i: \ol{A_i}\}_{i \in I} \sepr {\oplus}^\pri \{i: \ol{A_i}\}_{i \in I} \sepr \bullet \sepr \mu X.\ol{A} \sepr X
    \end{align*}
\end{definition}

\noindent
An endpoint of type `$A \tensor^\pri B$' (resp.\ `$A \parr^\pri B$') first outputs (resp.\ inputs) an endpoint of type $A$ and then behaves as $B$.
An endpoint of type `$\&^\pri \{i: A_i\}_{i \in I}$' offers a choice: after receiving a label $i \in I$, the endpoint behaves as $A_i$.
An endpoint of type `$\oplus^\pri \{i: A_i\}_{i \in I}$' selects a label $i \in I$ and then behaves as $A_i$.
An endpoint of type `$\bullet$' is closed; it does not require a priority, as closed endpoints do not exhibit behavior and thus are non-blocking.
We define `$\bullet$' as a single, self-dual type for closed endpoints, following Caires~\cite{report/Caires14}: the units `$\bot$' and `$\1$' of linear logic (used by, e.g., Caires and Pfenning~\cite{conf/concur/CairesP10} and Dardha and Gay \cite{conf/fossacs/DardhaG18} for session closing) are interchangeable in the absence of explicit closing.

Type `$\mu X.A$' denotes a recursive type, in which $A$ may contain occurrences of the recursion variable `$X$'.
As customary, `$\mu$' is a binder: it induces the standard notions of $\alpha$-equivalence, substitution (denoted `$A\subst{B/X}$'), and free recursion variables (denoted `$\frv(A)$').
We work with tail-recursive, contractive types, disallowing types of the form `$\mu X_1.\ldots.\mu X_n.X_1$'.
We adopt an equi-recursive view: a recursive type is equal to its unfolding.
We postpone formalizing the unfolding of recursive types, as it requires additional definitions to ensure consistency of priorities upon unfolding.

The priority of a type is determined by the priority of the type's outermost connective:

\begin{definition}[Priorities]\label{d:priority}
    For session type $A$, `$\pr(A)$' denotes its \emph{priority}:
    \begin{align*}
        \pr(A \tensor^\pri B)
        := \pr(A \parr^\pri B)
        &:= \pri
        &
        \pr(\mu X.A)
        &:= \pr(A)
        \\
        \pr(\oplus^\pri\{i:A_i\}_{i \in I})
        := \pr(\&^\pri\{i:A_i\}_{i \in I})
        &:= \pri
        &
        \pr(\bullet)
        := \pr(X)
        &:= \omega
    \end{align*}
\end{definition}

\noindent
The priority of `$\bullet$' and `$X$' is $\omega$: they denote ``final'', non-blocking actions of protocols.
Although `$\tensor$' and `$\oplus$' also denote non-blocking actions, their priority is not constant:
duality ensures that the priority for `$\tensor$' (resp.\ `$\oplus$') matches the priority of a corresponding `$\parr$' (resp.\ `$\&$'), which  denotes a blocking action.

Having defined the priority of types, we now turn to formalizing the unfolding of recursive types.
Recall the intuition that actions typed with lower priority should be performed before those with higher priority.
Based on this rationale, we observe that unfolding should increase the priorities of the unfolded type.
This is because the actions related to the unfolded recursion should be performed \emph{after} the prefix.
The following definition \emph{lifts} priorities in types:

\begin{definition}[Lift]\label{d:lift}
    For proposition $A$ and $t \in \mbb{N}$, we define `$\lift{t}A\mkern-2mu$' as the \emph{lift} operation:
    \begin{align*}
        \lift{t}(A \tensor^\pri B)
        &:= (\lift{t}A) \tensor^{\pri+t} (\lift{t}B)
        &
        \lift{t}(\oplus^\pri \{i: A_i\}_{i \in I})
        &:= \oplus^{\pri+t} \{i: \lift{t}A_i\}_{i \in I}
        &
        \lift{t}\bullet
        &:= \bullet
        \\
        \lift{t}(A \parr^\pri B)
        &:= (\lift{t}A) \parr^{\pri+t} (\lift{t}B)
        &
        \lift{t}(\&^\pri \{i: A_i\}_{i \in I})
        &:= \&^{\pri+t} \{i: \lift{t}A_i\}_{i \in I}
        \\
        \lift{t}(\mu X.A)
        &:= \mu X.\lift{t}(A)
        &
        \lift{t}X
        &:= X
    \end{align*}
\end{definition}

\noindent
Henceforth, the recursive type `$\mu X.A$' and its unfolding `$A\subst{\lift{t} \mu X.A/X}$' denote the same type, where the lift $t \in \mbb{N}$ of the unfolded recursive calls depends on the context in which the type appears.

\paragraph{Typing Rules}

The typing rules of APCP ensure that actions with lower priority are performed before those with higher priority (cf.\ Dardha and Gay~\cite{conf/fossacs/DardhaG18}).
To this end, they enforce the following laws:
\begin{enumerate}
    \item
        an action with priority $\pri$ must be prefixed only by inputs and branches with priority strictly smaller than $\pri$---this law does not hold for output and selection, as they are not prefixes;

    \item
        dual actions leading to synchronizations must have equal priorities (cf.\ Def.\ \labelcref{d:props}).
\end{enumerate}
Judgments are of the form `$P \vdash \Omega; \Gamma$', where $P$ is a process, $\Gamma$ is a  context that assigns types to channels (`$x{:}A$'), and $\Omega$ is a context that assigns natural numbers to recursion variables (`$X{:}n$').
The intuition behind the latter context is that it ensures the amount of context endpoints to concur between recursive definitions and calls.
Both contexts $\Gamma$ and $\Omega$ obey \emph{exchange}: assignments may be silently reordered.
$\Gamma$ is \emph{linear}, disallowing \emph{weakening} (i.e., all assignments must be used) and \emph{contraction} (i.e., assignments may not be duplicated).
$\Omega$ allows weakening and contraction, because a recursive definition does not necessarily require a recursive call although it may be called more than once.
The empty context is written `$\emptyset$'.
We write `$\pr(\Gamma)$' to denote the least priority of all types in $\Gamma$.
Notation `${(x_i{:}A_i)}_{i \in I}$' denotes indexing of assignments by $I$.
We write `$\lift{t} \Gamma$' to denote the component-wise extension of lift to typing contexts.

\begin{figure}[t]
    \begin{mdframed}
        {
            \mbox{}\hfill%
            \begin{wfit}
                \begin{prooftree}
                    \infAss{
                        $\raisebox{2.0ex}{}$
                    }
                    \infUn{
                        $\0 \vdash \Omega; \emptyset$
                    }{\scc{Empty}}
                \end{prooftree}
            \end{wfit}%
            \hfill%
            \begin{wfit}
                \begin{prooftree}
                    \infAss{
                        $P \vdash \Omega; \Gamma$
                    }
                    \infUn{
                        $P \vdash \Omega; \Gamma, x{:}\bullet$
                    }{$\bullet$}
                \end{prooftree}
            \end{wfit}%
            \hfill%
            \begin{wfit}
                \begin{prooftree}
                    \infAss{
                        $\raisebox{2.5ex}{}$
                    }
                    \infUn{
                        $x \fwd y \vdash \Omega; x{:}\ol{A}, y{:}A$
                    }{\scc{Id}}
                \end{prooftree}
            \end{wfit}%
            \hfill\mbox{}

            \smallskip
            \mbox{}\hfill%
            \begin{wfit}
                \begin{prooftree}
                    \infAss{
                        $P \vdash \Omega; \Gamma$
                        $\raisebox{1.8ex}{}$
                    }
                    \infAss{
                        $Q \vdash \Omega; \Delta$
                    }
                    \infBin{
                        $P \| Q \vdash \Omega; \Gamma, \Delta$
                    }{\scc{Mix}}
                \end{prooftree}
            \end{wfit}%
            \hfill%
            \begin{wfit}
                \begin{prooftree}
                    \infAss{
                        $P \vdash \Omega; \Gamma, x{:}A, y{:}\ol{A}$
                    }
                    \infUn{
                        $\nu{x y} P \vdash \Omega; \Gamma$
                    }{\scc{Cycle}}
                \end{prooftree}
            \end{wfit}%
            \hfill\mbox{}

            \smallskip
            \mbox{}\hfill%
            \begin{wfit}
                \begin{prooftree}
                    \infAss{
                        $\raisebox{2.3ex}{}$
                    }
                    \infUn{
                        $x[y,z] \vdash \Omega; x{:}A \tensor^\pri B, y{:}\ol{A}, z{:}\ol{B}$
                    }{$\tensor$}
                \end{prooftree}
            \end{wfit}%
            \hfill%
            \begin{wfit}
                \begin{prooftree}
                    \infAss{
                        $P \vdash \Omega; \Gamma, y{:}A, z{:}B$
                    }
                    \infAss{
                        $\pri < \pr(\Gamma)$
                    }
                    \infBin{
                        $x(y, z); P \vdash \Omega; \Gamma, x{:}A \parr^\pri B$
                    }{$\parr$}
                \end{prooftree}
            \end{wfit}%
            \hfill\mbox{}

            \smallskip
            \mbox{}\hfill%
            \begin{wfit}
                \begin{prooftree}
                    \infAss{
                        $j \in I$
                        $\raisebox{1.8ex}{}$
                    }
                    \infUn{
                        $x[z] \triangleleft j \vdash \Omega; x{:}{\oplus}^\pri\{i: A_i\}_{i \in I}, z{:}\ol{A_j}$
                    }{$\oplus$}
                \end{prooftree}
            \end{wfit}%
            \hfill%
            \begin{wfit}
                \begin{prooftree}
                    \infAss{
                        $\forall i \in I.~ P_i \vdash \Omega; \Gamma, z{:}A_i$
                    }
                    \infAss{
                        $\pri < \pr(\Gamma)$
                    }
                    \infBin{
                        $x(z) \triangleright \{i: P_i\}_{i \in I} \vdash \Omega; \Gamma, x{:}\&^\pri\{i: A_i\}_{i \in I}$
                    }{$\&$}
                \end{prooftree}
            \end{wfit}%
            \hfill\mbox{}

            \smallskip
            \mbox{}\hfill%
            \begin{wfit}
                \begin{prooftree}
                    \infAss{
                        $P \vdash \Omega, X{:}|I|; {(x_i{:}A_i)}_{i \in I}$
                    }
                    \infAss{
                        $\forall i \in I.~ A_i \neq X$
                    }
                    \infBin{
                        $\mu X({(x_i)}_{i \in I}); P \vdash \Omega; {(x_i{:}\mu X.A_i)}_{i \in I}$
                    }{\scc{Rec}}
                \end{prooftree}
            \end{wfit}%
            \hfill%
            \begin{wfit}
                \begin{prooftree}
                    \infAss{
                        $\raisebox{2.5ex}{}$
                    }
                    \infUn{
                        $X\call{{(x_i)}_{i \in I}} \vdash \Omega, X{:}|I|; {(x_i{:}X)}_{i \in I}$
                    }{\scc{Var}}
                \end{prooftree}
            \end{wfit}%
            \hfill\mbox{}

            \medskip
            \hbox to \textwidth{\leaders\hbox to 3pt{\hss . \hss}\hfil}

            \medskip
            \mbox{}\hfill%
            \begin{wfit}
                \begin{prooftree}
                    \infAss{
                        $P \vdash \Omega; \Gamma, y{:}A, x{:}B$
                    }
                    \infUn{
                        $\ol{x}[y] \cdot P \vdash \Omega; \Gamma, x{:}A \tensor^\pri B$
                    }{$\tensor^\star$}
                \end{prooftree}
            \end{wfit}%
            \hfill%
            \begin{wfit}
                \begin{prooftree}
                    \infAss{
                        $P \vdash \Omega; \Gamma, x{:}A_j$
                    }
                    \infAss{
                        $j \in I$
                    }
                    \infBin{
                        $\ol{x} \triangleleft j \cdot P \vdash \Omega; \Gamma, x{:}{\oplus}^\pri\{i: A_i\}_{i \in I}$
                    }{$\oplus^\star$}
                \end{prooftree}
            \end{wfit}%
            \hfill%
            \begin{wfit}
                \begin{prooftree}
                    \infAss{
                        $P \vdash \Omega; \Gamma$
                    }
                    \infAss{
                        $t \in \mbb{N}$
                    }
                    \infBin{
                        $P \vdash \Omega; \lift{t} \Gamma$
                    }{\scc{Lift}}
                \end{prooftree}
            \end{wfit}%
            \hfill\mbox{}%
        }%
    \end{mdframed}

    \caption{The typing rules of APCP (top) and admissible rules (bottom).}
    \label{f:apcpInf}
\end{figure}

\Cref{f:apcpInf} (top) gives the typing rules.
Typing is closed under structural congruence; we sometimes use this explicitly in typing derivations in the form of a rule `$\equiv$'.
Axiom `\scc{Empty}' types an inactive process with no endpoints.
Rule `$\bullet$' silently adds a closed endpoint to the typing context.
Axiom `\scc{Id}' types forwarding between endpoints of dual type.
Rule `\scc{Mix}' types the parallel composition of two processes that do not share assignments on the same endpoints.
Rule `\scc{Cycle}' removes two endpoints of dual type from the context by adding a restriction on them.
Note that a single application of `\scc{Mix}' followed by `\scc{Cycle}' coincides with the usual rule `\scc{Cut}' in type systems based on linear logic~\cite{conf/concur/CairesP10,conf/icfp/Wadler12}.
Axiom `$\tensor$' types an output action; this rule does not have premises to provide a continuation process, leaving the free endpoints to be bound to a continuation process using `\scc{Mix}' and `\scc{Cycle}'.
Similarly, axiom `$\oplus$' types an unbounded selection action.
Priority checks are confined to rules `$\parr$' and `$\&$', which type an input and a branching prefix, respectively.
In both cases, the used endpoint's priority must be lower than the priorities of the other types in the continuation's typing context.

Rule `\scc{Rec}' types a recursive definition by eliminating a recursion variable from the recursion context whose value concurs with the size of the typing context, where contractiveness is guaranteed by requiring that the eliminated recursion variable may not appear unguarded in each of the context's types.
Axiom `\scc{Var}' types a recursive call by adding a recursion variable to the context with the amount of introduced endpoints.
As mentioned before, the value of the introduced and consequently eliminated recursion variable is crucial in ensuring that a recursion is called with the same amount of channels as required by its definition.

Let us compare our typing system to that of Dardha and Gay~\cite{conf/fossacs/DardhaG18} and DeYoung \etal~\cite{conf/csl/DeYoungCPT12}.
Besides our support for recursion, the main difference is that our rules for output and selection are axioms.
This makes priority checking much simpler for APCP than for Dardha and Gay's PCP: our outputs and selections have no typing context to check priorities against, and types for closed endpoints have no priority at all.
Although DeYoung \etal's output and selection actions are atomic too, their corresponding rules are similar to the rules of Dardha and Gay: the rules require continuation processes as premises, immediately binding the sent endpoints.

As anticipated, the binding of output and selection actions to continuation processes (\Cref{n:sugar}) is derivable in APCP.
The corresponding typing rules in \Cref{f:apcpInf} (bottom) are admissible using `\scc{Mix}' and `\scc{Cycle}'.
Note that it is not necessary to include rules for the sugared input and branching in \Cref{n:sugar}, because they rely on name substitution only and typing is closed under structural congruence and thus name substitution.
\Cref{f:apcpInf} (bottom) also includes an admissible rule `\scc{Lift}' that lifts a process' priorities.

\begin{theorem}\label{t:admissible}
    The rules `$\tensor^\star$', `$\oplus^\star$', and `\scc{Lift}' in \Cref{f:apcpInf} (bottom) are admissible.
\end{theorem}

\begin{figure}
    \begin{mdframed}
        \mbox{}\hfill%
        \begin{wfit}
            \begin{prooftree}
                \infAss{
                    $P \vdash \Gamma, y{:}A, x{:}B$
                }
                \infUn{
                    $\ol{x}[y] \cdot P \vdash \Gamma, x{:}A \tensor^\pri B$
                }{$\tensor^\star$}
            \end{prooftree}
        \end{wfit}%
        \hfill%
        $\Rightarrow$%
        \hfill%
        \begin{wfit}
            \begin{prooftree}
                \infAx{
                    $x[a, b] \vdash x{:}A \tensor^\pri B, a{:}\ol{A}, b{:}\ol{B}$
                }{$\tensor$}
                \infAss{
                    $P \vdash \Gamma, y{:}A, x{:}B$
                }
                \infUn{
                    $P\subst{z/x} \vdash \Gamma, y{:}A, z{:}B$
                }{$\equiv$}
                \infBin{
                    $x[a, b] \| P\subst{z/x} \vdash \Gamma, x{:}A \tensor^\pri B, y{:}A, a{:}\ol{A}, z{:}B, b{:}\ol{B}$
                }{\scc{Mix}}
                \infDblUn{
                    $\underbrace{\nu{y a} \nu{z b} (x[a, b] \| P\subst{z/x}) }_{\ol{x}[y] \cdot P ~\text{(cf.\ \Cref{n:sugar})}} \vdash \Gamma, x{:}A \tensor^\pri B$
                }{\scc{Cycle}\textsuperscript{2}}
            \end{prooftree}
        \end{wfit}%
        \hfill\mbox{}

        \smallskip
        \mbox{}\hfill%
        \begin{wfit}
            \begin{prooftree}
                \infAss{
                    $P \vdash \Gamma, x{:}A_j$
                }
                \infAss{
                    $j \in I$
                }
                \infBin{
                    $\ol{x} \triangleleft j \cdot P \vdash \Gamma, x{:}{\oplus}^\pri \{i: A_i \}_{i \in I}$
                }{$\oplus^\star$}
            \end{prooftree}
        \end{wfit}%
        \hfill%
        $\Rightarrow$%
        \hfill%
        \begin{wfit}
            \begin{prooftree}
                \infAss{
                    $j \in I$
                }
                \infUn{
                    $x[b] \triangleleft j \vdash x{:}{\oplus}^\pri \{i: A_i \}_{i \in I}, b{:}\ol{A_j}$
                }{$\oplus$}
                \infAss{
                    $P \vdash \Gamma, x{:}A_j$
                }
                \infUn{
                    $P \subst{z/x} \vdash \Gamma, z{:}A_j$
                }{$\equiv$}
                \infBin{
                    $x[b] \triangleleft j \| P \subst{z/x} \vdash \Gamma, x{:}{\oplus}^\pri \{i: A_i \}_{i \in I}, z{:}A_j, b{:}\ol{A_j}$
                }{\scc{Mix}}
                \infUn{
                    $\underbrace{\nu{z b} (x[b] \triangleleft j \| P \subst{z/x})}_{\ol{x} \triangleleft j \cdot P ~\text{(cf.\ \Cref{n:sugar})}} \vdash \Gamma, x{:}{\oplus}^\pri \{i: A_i \}_{i \in I}$
                }{\scc{Cycle}}
            \end{prooftree}
        \end{wfit}%
        \hfill\mbox{}
    \end{mdframed}

    \caption{Proof that rules `$\tensor^\star$' and `$\oplus^\star$' are admissible (cf.\ \Cref{t:admissible}).}
    \label{f:admissible}
\end{figure}

\begin{proof}
    We show the admissibility of rules $\tensor^\star$ and $\oplus^\star$ by giving their derivations in \Cref{f:admissible} (omitting the recursion context).
    The rule `\scc{Lift}' is admissible, because $P \vdash \Omega; \Gamma$ implies $P \vdash \Omega; \lift{t} \Gamma$ (cf.\ Dardha and Gay~\cite{conf/fossacs/DardhaG18}), by simply increasing all priorities in the derivation of $P$ by $t$.
\end{proof}

\noindent
\Cref{t:admissible} highlights how APCP's asynchrony uncovers a more primitive, lower-level view of message-passing.
In the next subsection we discuss deadlock freedom, which follows from a correspondence between reduction and the removal of `\scc{Cycle}' rules from typing derivations.
In the case of APCP, this requires care: binding output and selection actions to continuation processes leads to applications of `\scc{Cycle}' not immediately corresponding to reductions.

\subsection{Type Preservation and Deadlock Freedom}
\label{ss:results}

Well-typed processes satisfy protocol fidelity, communication safety, and deadlock freedom.
All these properties follow from \emph{type preservation} (also known as \emph{subject reduction}), which ensures that reduction preserves typing.
In contrast to Caires and Pfenning~\cite{conf/concur/CairesP10} and Wadler~\cite{conf/icfp/Wadler12}, where type preservation corresponds to the elimination of (top-level) applications of rule \scc{Cut}, in APCP it corresponds to the elimination of (top-level) applications of rule \scc{Cycle}.

\begin{theorem}[Type Preservation]\label{t:subjectRed}
    If \,$P \vdash \Omega; \Gamma$ and $P \redd Q$, then $Q \vdash \Omega; \lift{t}\Gamma$ for $t \in \mbb{N}$.
\end{theorem}

\begin{figure}[t]
    \begin{mdframed}
        {
            \mbox{}\hfill%
            \begin{wfit}
                \begin{prooftree}
                    \infAx{
                        $x[a,b] \vdash x{:}A \tensor^\pri B, a{:}\ol{A}, b{:}\ol{B}$
                    }{$\tensor$}
                    \infAss{
                        $P \vdash \Gamma, v{:}\ol{A}, z{:}\ol{B}$
                    }
                    \infAss{
                    }
                    \infBin{
                        $y(v,z).P \vdash \Gamma, y{:}\ol{A} \parr^\pri \ol{B}$
                    }{$\parr$}
                    \infDblBin{
                        $\nu{x y}(x[a,b] \| y(v,z).P) \vdash \Gamma, a{:}\ol{A}, b{:}\ol{B}$
                    }{$\mkern-10mu\begin{array}{c}
                        \\ \scc{Mix}\,+ \\ \scc{Cycle}
                    \end{array}$}
                \end{prooftree}
            \end{wfit}%
            \hfill%
            $\redd$%
            \hfill%
            \begin{wfit}
                \begin{prooftree}
                    \infAss{
                        $P \vdash \Gamma, v{:}\ol{A}, z{:}\ol{B}$
                    }
                    \infUn{
                        $P \subst{a/v,b/z} \vdash \Gamma, a{:}\ol{A}, b{:}\ol{B}$
                    }{$\equiv$}
                \end{prooftree}
            \end{wfit}%
            \hfill\mbox{}%

            \medskip
            \hbox to \textwidth{\leaders\hbox to 3pt{\hss . \hss}\hfil}

            \smallskip
            Below, the contexts $\Gamma'$ and $\Delta'$ together contain $\tilde{v}$ and $\tilde{w}$, i.e.\ $\Gamma', \Delta' = {(v_i{:}C_i)}_{v_i \in \tilde{v}}, {(w_i{:}\ol{C_i})}_{w_i \in \tilde{w}}$.

            \begin{prooftree}
                \infAss{
                    $P \vdash \Gamma, \Gamma', y{:}A, z{:}B$
                }
                \infAss{
                    $\pri < \pr(\Gamma)$
                }
                \infBin{
                    $x(y,z).P \vdash \Gamma, \Gamma', x{:}A \parr^\pri B$
                }{$\parr$}
                \infAss{
                    $Q \vdash \Delta, \Delta'$
                }
                \infDblBin{
                    $\nu{\tilde{v} \tilde{w}}(x(y,z).P \| Q) \vdash \Gamma, \Delta, x{:}A \parr^\pri B$
                }{$\scc{Mix}+\scc{Cycle}^\ast$}
            \end{prooftree}

            \mbox{}\hfill%
            $\redd$%
            \hfill\mbox{}

            \begin{prooftree}
                \infAss{
                    $P \vdash \Gamma, \Gamma', y{:}A, z{:}B$
                }
                \infAss{
                    $Q \vdash \Delta, \Delta'$
                }
                \infDblBin{
                    $\nu{\tilde{v} \tilde{w}}(P \| Q) \vdash \Gamma, \Delta, y{:}A, z{:}B$
                }{$\scc{Mix}+\scc{Cycle}^\ast$}
                \infUn{
                    $\nu{\tilde{v} \tilde{w}}(P \| Q) \vdash \lift{\pri+1} \Gamma, \lift{\pri+1} \Delta, y{:}\lift{\pri+1} A, z{:}\lift{\pri+1} B$
                }{\scc{Lift}}
                \infAss{
                    $\pri < \pr(\lift{\pri + 1} \Gamma, \lift{\pri+1} \Delta)$
                }
                \infBin{
                    $x(y,z).\nu{\tilde{v} \tilde{w}}(P \| Q) \vdash \lift{\pri+1} \Gamma, \lift{\pri+1} \Delta, x{:}(\lift{\pri+1} A) \parr^\pri (\lift{\pri+1} B)$
                }{$\parr$}
            \end{prooftree}
        }%
    \end{mdframed}

    \caption{Type Preservation (cf.\ \Cref{t:subjectRed}) in rules $\protect\brred{\protect\tensor\protect\parr}$ (top) and $\protect\krred{\protect\parr}$ (bottom).}
    \label{f:subRed}
\end{figure}

\begin{proof}
    By induction on the reduction $\redd$, analyzing the last applied rule (\figref{f:procdef} (bottom)).
    The cases of the closure rules $\rred{\equiv}$, $\rred{\onu}$, and $\rred{\|}$ easily follow from the IH.
    The key cases are the $\beta$- and $\kappa$-rules.
    \Cref{f:subRed} shows two representative instances (eluding the recursion context $\Omega$): rule $\brred{\tensor\parr}$ (top), a synchronization, and rule $\krred{\parr}$ (bottom), a commuting conversion.
    Note how, in the case of rule $\kappa_{\parr}$, the lift $\lift{t}$ ensures consistent priority checks.
\end{proof}

Protocol fidelity ensures that processes respect their intended (session) protocols.
Communication safety ensures the absence of communication errors and mismatches in processes.
Correct typability gives a static guarantee that a process conforms to its ascribed session protocols; type preservation gives a dynamic guarantee.
Because session types describe the intended protocols and error-free exchanges, type preservation entails both protocol fidelity and communication safety.
We refer the curious reader to the early work by Honda \etal~\cite{conf/esop/HondaVK98} for a detailed account, which shows by contradiction that well-typed processes do not reduce to so-called error processes.
This is a well-known and well-understood result.

In what follows, we consider a process to be deadlocked if it is not the inactive process and cannot reduce.
Our deadlock freedom result for APCP adapts that for PCP~\cite{conf/fossacs/DardhaG18}, which involves three steps:
\begin{enumerate}
    \item
        First, \scc{Cycle}-elimination states that we can remove all applications of \scc{Cycle} in a typing derivation without affecting the derivation's assumptions and conclusion.

    \item
        Only the removal of \emph{top-level} \scc{Cycle}s captures the intended process semantics, as the removal of other \scc{Cycle}s corresponds to reductions behind prefixes which is not allowed~\cite{conf/icfp/Wadler12,conf/fossacs/DardhaG18}.
        Therefore, the second step is \emph{top-level deadlock freedom}, which states that a process with a top-level \scc{Cycle} reduces until there are no top-level \scc{Cycle}s left.

    \item
        Third, deadlock freedom follows for processes typable under empty contexts.
\end{enumerate}
Here, we address cycle-elimination and top-level deadlock-freedom in one proof.

As mentioned before, binding APCP's asynchronous outputs and selections to continuations involves additional, low-level uses of \scc{Cycle}, which we cannot eliminate through process reduction.
Therefore, we establish top-level deadlock freedom for \emph{live processes} (\Cref{t:dlFree}).
A process is live if it is equivalent to a restriction on \emph{active names} that perform unguarded actions.
This way, e.g., in `$x[y,z]$' the name $x$ is active, but $y$ and $z$ are not.

\begin{definition}[Active Names]\label{d:an}
    The \emph{set of active names} of $P$, denoted `$\an(P)\mkern-2mu$', contains the (free) names that are used for unguarded actions (output, input, selection, branching):

    \vspace{-2em}
    \begin{align*}
        \an(x[y,z])
        & := \{x\}
        &
        \an(x(y,z).P)
        & := \{x\}
        &
        \an(\0)
        & := \emptyset
        \\
        \an(x[z] \triangleleft j)
        & := \{x\}
        &
        \an(x(z) \triangleright \{i: P_i\}_{i \in I})
        & := \{x\}
        &
        \an(x \fwd y)
        & := \{x,y\}
        \\
        \an(P \| Q)
        & := \an(P) \cup \an(Q)
        &
        \an(\mu X(\tilde{x}); P)
        & := \an(P)
        \\
        \an(\nu{x y}P)
        & := \an(P) \setminus \{x,y\}
        &
        \an(X\call{\tilde{x}})
        & := \emptyset
    \end{align*}
\end{definition}

\begin{definition}[Live Process]\label{d:live}
    A process $P$ is \emph{live}, denoted `$\mkern1mu\live(P)\mkern-2mu$', if there are names $x,y$ and process $P'$ such that $P \equiv \nu{x y}P'$ with $x,y \in \an(P')$.
\end{definition}

We additionally need to account for recursion: as recursive definitions do not entail reductions, we must fully unfold them before eliminating \scc{Cycle}s.

\begin{lemma}[Unfolding]\label{l:unfold}
    If $P \vdash \Omega; \Gamma$, then there is process $P^\star$ such that $P^\star \equiv P$ and $P^\star$ is not of the form `$\mu X(\tilde{x}); Q$' and $P^\star \vdash \Omega; \Gamma$.
\end{lemma}

\begin{figure}[t]
    \begin{mdframed}
        \begin{prooftree}
            \infAx{
                $X\call{{(y_i)}_{i \in I}} \vdash \Omega', X{:}|I|; {(y_i{:}X)}_{i \in I}$
            }{\scc{Var}} \noLine
            \def\extraVskip{-2.5pt}
            \UnaryInfC{
                $\raisebox{2pt}{$\vdots$}$
            }
            \infAss{
                $\raisebox{2pt}{$\iddots$}$
            }
            \def\extraVskip{2pt}
            \infDblBin{
                $Q \vdash \Omega, X{:}|I|; {(x_i{:}A_i)}_{i \in I}$
            }{\ldots}
            \infUn{
                $\mu X({(x_i)}_{i \in I}); Q \vdash \Omega; {(x_i{:}\mu X.A_i)}_{i \in I}$
            }{\scc{Rec}}
        \end{prooftree}

        \medskip
        \hbox to \textwidth{\leaders\hbox to 3pt{\hss . \hss}\hfil}

        \begin{prooftree}
            \infAx{
                $X\call{{(y_i)}_{i \in I}} \vdash \Omega'', X{:}|I|; {(y_i{:}X)}_{i \in I}$
            }{\scc{Var}} \noLine
            \def\extraVskip{-2.5pt}
            \UnaryInfC{
                $\raisebox{2pt}{$\vdots$}$
            }
            \infAss{
                $\raisebox{2pt}{$\iddots$}$
            }
            \def\extraVskip{2pt}
            \infDblBin{
                $Q \vdash \Omega', X{:}|I|; {(x_i{:}A_i)}_{i \in I}$
            }{\ldots}
            \infUn{
                $Q \subst{{(y_i)}_{i \in I}/{(x_i)}_{i \in I}} \vdash \Omega', X{:}|I|; {(y_i{:}A_i)}_{i \in I}$
            }{$\equiv$}
            \infUn{
                $\mu X({(y_i)}_{i \in I}); Q \subst{{(y_i)}_{i \in I}/{(x_i)}_{i \in I}} \vdash \Omega'; {(y_i{:}\mu X.A_i)}_{i \in I}$
            }{\scc{Rec}}
            \infAss{
                $t \geq \max_{\sff{pr}}{(A_i)}_{i \in I}$
            }
            \infBin{
                $\mu X({(y_i)}_{i \in I}); Q \subst{{(y_i)}_{i \in I}/{(x_i)}_{i \in I}} \vdash \Omega'; {(y_i{:}\lift{t} \mu X.A_i)}_{i \in I}$
            }{\scc{Lift}} \noLine
            \def\extraVskip{-2.5pt}
            \UnaryInfC{
                $\raisebox{2pt}{$\vdots$}$
            }
            \infAss{
                $\raisebox{2pt}{$\iddots$}$
            }
            \def\extraVskip{2pt}
            \infDblBin{
                $R \vdash \Omega; {(x_i{:}A_i\subst{\lift{t} \mu X.A_i/X})}_{i \in I}$
            }{\ldots}
        \end{prooftree}
    \end{mdframed}

    \caption{Typing recursion before (top) and after (bottom) unfolding (cf.\ \Cref{l:unfold}).}
    \label{f:unfold}
\end{figure}

\begin{proof}
    By induction on the amount $n$ of consecutive recursive definitions prefixing $P$, such that $P$ is of the form `$\mkern-1mu\mu X_1(\tilde{x}); \ldots; \mu X_n(\tilde{x}); Q$'.
    If $n = 0$, the thesis follows immediately by letting $P^\star := P$.

    Otherwise, $n \geq 1$.
    Then there are $X,Q$ such that $P = \mu X({(x_i)}_{i \in I}); Q$.
    By inversion of typing rule \scc{Rec}, $P \vdash \Omega; {(x_i{:}\mu X.A_i)}_{i \in I}$.
    Generally speaking, such typing derivations have the shape as in \Cref{f:unfold} (top), with zero or more \scc{Var}-axioms on $X$ appearing at the top.
    We use structural congruence (\figref{f:procdef} (middle)) to unfold the recursion in $P$, obtaining the process $R := Q \rsubst{\mu X({(y_i)}_{i \in I}); Q \subst{{(y_i)}_{i \in I} / {(x_i)}_{i \in I}} / X\call{{(y_i)}_{i \in I}}} \equiv P$.

    We can type $R$ by taking the derivation of $P$ (cf.\ \Cref{f:unfold} (top)), removing the final application of the \scc{Rec}-rule and replacing any uses of the \scc{Var}-axiom on $X$ by a copy of the original derivation, applying $\alpha$-conversion where necessary.
    Moreover, we lift the priorities of all types by at least the highest priority occurring in any type in $\Gamma$ using the \scc{Lift}-rule, ensuring that priority conditions on typing rules remain valid; we explicitly use \emph{at least} the highest priority, as the context of connected endpoints may lift the priorities in dual types even more.
    Writing the highest priority in $\Gamma$ as `$\max_{\sff{pr}}(\Gamma)$', the resulting proof is of the shape in \Cref{f:unfold} (bottom).
    Since types are equi-recursive, $A_i\subst{\lift{t} \mu X.A_i/X} = A_i$ for every $i \in I$.
    Hence, ${(y_i{:}A_i\subst{\lift{t} \mu X.A_i / X})}_{i \in I} = \Gamma$.
    Thus, the above is a valid derivation of $R \vdash \Omega; \Gamma$.

    The rules applied after \scc{Lift} in the derivation of $R$ in \Cref{f:unfold} (bottom) are the same as those applied after \scc{Var} and before \scc{Rec} in the derivation of $P$ in \Cref{f:unfold} (top) before unfolding.
    By the assumption that recursion is contractive, there must be an application of a rule other than \scc{Rec} in this part of the derivation.
    Therefore, the application of \scc{Rec} in the derivation of $R$ is not part of a possible sequence of \scc{Rec}s in the last-applied rules of this derivation.
    Hence, since we removed the final application of \scc{Rec} in the derivation of $P$, the size of this sequence of \scc{Rec}s is $n-1$, i.e.\ $R$ is prefixed by $n-1$ recursive definitions.
    Thus, we apply the IH to find a process $P^\star$ not prefixed by recursive definitions s.t.\ $P^\star \equiv R \equiv P \vdash \Omega; \Gamma$.
\end{proof}

Dardha and Gay's top-level deadlock freedom result concerns a sequence of reduction steps that reaches a process that is not live anymore~\cite{conf/fossacs/DardhaG18}.
In our case, top-level deadlock freedom concerns a single reduction step only, because recursive processes might stay live across reductions forever.

\begin{theorem}[Top-Level Deadlock Freedom]\label{t:dlFree}
    If $P \vdash \emptyset; \Gamma$ and $\live(P)$, then there is process $Q$ such that $P \redd Q$.
\end{theorem}

\begin{proof}
    By structural congruence (\figref{f:procdef} (middle)), there is $P_c = \nu{x_i y_i}_{i \in I} \nu{\tilde{n} \tilde{m}} P_m$ such that $P_c \equiv P$, with $P_m = \prod_{k \in K} P_k$ and $P_m \vdash \emptyset; \Lambda, {(x_i{:}A_i, y_i{:}\ol{A_i})}_{i \in I}$ s.t.\ for every $i \in I$, $x_i$ and $y_i$ are active names in $P_m$, and $\Lambda$ consists of $\Gamma$ and the channels $\tilde{n},\tilde{m}$ which are dually typed pairs of endpoints of which at least one is inactive in $P_m$.
    Because $P$ is live, there is always at least one pair $x_i,y_i$.

    Next, we take the $j \in I$ s.t.\ $A_j$ has the least priority, i.e.\ $\forall i \in I \setminus \{j\}.~\pr(A_j) \leq \pr(A_i)$.
    If there are multiple to choose from, any suffices.
    The rest of the analysis depends on whether there is an endpoint $z$ of input/branching type in $\Gamma$ with lower priority than $\pr(A_j)$.
    We thus distinguish the two cases below.
    Note that output/selection types in $\Gamma$ are associated with non-blocking actions and can be safely ignored.
    \begin{itemize}
        \item
            If there is such $z$, assume w.l.o.g.\ it is of input type.
            The input on $z$ cannot be prefixed by an input/branch on another endpoint, because then that other endpoint would have a type with lower priority than $z$.
            Hence, there is $k' \in K$ s.t.\ $P_{k'} = z(u,v); P'_{k'}$.
            We thus apply communicating conversion $\krred{\parr}$ to find $Q$ such that $P \redd Q$:
            \[
                P \equiv \nu{x_i y_i}_{i \in I} \nu{\tilde{n} \tilde{m}} (\prod_{k \in K \setminus \{k'\}} P_k \| z(u,v); P'_{k'}) \redd z(u,v); \nu{x_i y_i}_{i \in I} \nu{\tilde{n} \tilde{m}} (\prod_{k \in K \setminus \{k'\}} P_k \| P'_{k'}) = Q
            \]

        \item
            If there is no such $z$, we continue with $x_j{:}A_j$ and $y_j{:}\ol{A_j}$.
            In case there is $k' \in K$ s.t.\ $P_k \equiv u \fwd v$ with $u \in \{x_j,y_j\}$, the reduction is trivial by $\rred{\scc{Id}}$; we w.l.o.g.\ assume there is no such $k'$.

            By duality, $A_j$ and $\ol{A_j}$ have the same priority, so priority checks in typing derivations prevent an input/branching prefix on $x_j$ (resp.\ $y_j$) from blocking an output/selection on $y_j$ (resp.\ $x_j$).
            Hence, $x_j$ and $y_j$ appear in separate parallel components of $P_m$, i.e.\ $P_m = P_{x_j} \| P_{y_j} \| P_R$ s.t.\
            \[
                P_{x_j} \vdash \emptyset; \Lambda_{x_j}, x_j{:}A_j ~,
                \quad
                P_{y_j} \vdash \emptyset; \Lambda_{y_j}, y_j{:}\ol{A_j} ~\text{, and}
                \quad
                P_R \vdash \emptyset; \Lambda_R ~,
            \]
            where $\Lambda_{x_j}, \Lambda_{y_j}, \Lambda_R, x_j{:}A_j, y_j{:}\ol{A_j} = \Lambda, {(x_i{:}A_i, y_i{:}\ol{A_i})}_{i \in I}$.

            By \Cref{l:unfold} (unfolding), $P_{x_j} \equiv P_{x_j}^\star$ and $P_{y_j} \equiv P_{y_j}^\star$ s.t.\ $P_{x_j}^\star$ and $P_{y_j}^\star$ are not prefixed by recursive definitions and $P_{x_j}^\star \vdash \emptyset; \Lambda_{x_j}, x_j{:}A_j$ and $P_{y_j}^\star \vdash \emptyset; \Lambda_{y_j}, y_j{:}\ol{A_j}$.
            We take the unfolded form of $A_j$: by the contractiveness of recursive types, $A_j$ has at least one connective.
            We w.l.o.g.\ assume that $A_j$ is an input or branching type, i.e.\ either (a) $A_j = B \parr^\pri C$ or (b) $A_j = \&^\pri \{l:B_l\}_{l \in L}$.

            Since $\pr(A_j) = \pri$ is the least of the priorities in $\Gamma$, we know that either (in case a) $P_{x_j}^\star \equiv x_j(v,z).Q_{x_j}$ or (in case b) $P_{x_j}^\star \equiv x_j(z) \triangleright \{l:Q_{x_j}^l\}_{l \in L}$.
            Moreover, since either (in case a) $\ol{A_j} = \ol{B} \tensor^\pri \ol{C}$ or (in case b) $\ol{A_j} = {\oplus}^\pri \{l:\ol{B_l}\}_{l \in L}$, we have that either (in case a) $P_{y_j}^\star \equiv y_j[a,b] \| Q_{y_j}$ or (in case b) $P_{y_j}^\star \equiv y_j[b] \triangleleft l^\star \| Q_{y_j}$ for $l^\star \in L$.
            In case (a), let $Q'_{x_j} := Q_{x_j} \subst{a/v,b/z}$; in case (b), let $Q'_{x_j} := Q_{x_j}^{l^\star} \subst{b/z}$.
            Then, (in case a) by reduction $\brred{\tensor\parr}$ or (in case b) by reduction $\brred{\oplus\&}$,
            \begin{align*}
                P \equiv \nu{x_i y_i}_i \nu{\tilde{n} \tilde{m}} (P_{x_j}^\star \| P_{y_j}^\star \| P_R)
                & \equiv \nu{x_i y_i}_{i \setminus j} \nu{\tilde{n} \tilde{m}} (\nu{x_j y_j}(P_{x_j}^\star \| P_{y_j}^\star) \| P_R)
                \\
                &
                \redd \nu{x_i y_i}_{i \setminus j} \nu{\tilde{n} \tilde{m}} (Q'_{x_j} \| Q_{y_j} \| P_R).
                \tag*{\qedhere}
            \end{align*}
    \end{itemize}
\end{proof}

Our deadlock freedom result concerns processes typable under empty contexts (as in, e.g., Caires and Pfenning~\cite{conf/concur/CairesP10} and Dardha and Gay~\cite{conf/fossacs/DardhaG18}).
This way, the reduction guaranteed by \Cref{t:dlFree} corresponds to a synchronization ($\beta$-rule), rather than a commuting conversion ($\kappa$-rule).
We first need a lemma which ensures that non-live processes typable under empty contexts do not contain actions or prefixes.

\begin{lemma}\label{l:nLivenAction}
    If $P \vdash \emptyset; \emptyset$ and $P$ is not live, then $P$ contains no actions or prefixes whatsoever.
\end{lemma}

 \begin{proof}
    Suppose, for contradiction, that $P$ does contain actions or prefixes.
    For example, $P$ contains some subterm $x(y,z); P'$.
    Because $P \vdash \emptyset; \emptyset$, there must be a restriction on $x$ in $P$ binding it with, e.g., $x'$.
    Now, $x'$ does not appear in $P'$, because the type of $x$ in the derivation of $P \vdash \emptyset; \emptyset$ must be lower than the types of the endpoints in $P'$, and by duality the types of $x$ and $x'$ have equal priority.
    Hence, there is some $Q$ s.t. $P \equiv \nu{\tilde{u}\tilde{v}}\nu{xx'}(x(y,z); P' \| Q)$ where $x' \in \fn(Q)$.
    There are two cases for the appearance of $x'$ in $Q$: (1) not prefixed, or (2) prefixed.
    \begin{itemize}
        \item
            In case (1), $x' \in \an(Q)$, so the restriction on $x,x'$ in $P$ is on a pair of active names, contradicting the fact that $P$ is not live.

        \item
            In case (2), $x'$ appears in $Q$ behind at least one prefix.
            For example, $Q$ contains some subterm $a(b,c); Q'$ where $x' \in \fn(Q')$.
            Again, $a$ must be bound in $P$ to, e.g., $a'$.
            Through similar reasoning as above, we know that $a'$ does not appear in $Q'$.
            Moreover, the type of $a$ must have lower priority than the type of $x'$, so by duality the type of $a'$ must have lower priority than the type of $x$.
            So, $a'$ also does not appear in $P'$.
            Hence, there is $R$ s.t. $P \equiv \nu{\tilde{u}\tilde{v}}\nu{aa'}\nu{xx'}(x(y,z); P' \| a(b,c); Q' \mid R)$ where $a' \in \fn(R)$.

            Now, the case split on whether $a'$ appears prefixed in $R$ or not repeats, possibly finding new names that prefix the current name again and again following case (2).
            However, process terms are finite in size, so we know that at some point there cannot be an additional parallel component in $P$ to bind the new name, contradicting the existence of the newly found prefix.
            Hence, eventually case (1) will be reached, uncovering a restriction on a pair of active names and contradicting the fact that $P$ is not live.
    \end{itemize}
    In conclusion, the assumption that there are actions or prefixes in $P$ leads to a contradiction.
    Hence, $P$ contains no actions or prefixes whatsoever.
\end{proof}

We now state our deadlock freedom result:

\begin{theorem}[Deadlock Freedom]\label{t:closedDLFree}
    If $P \vdash \emptyset; \emptyset$, then either $P \equiv \0$ or $P \redd_\beta Q\mkern2mu$ for some $Q$.
\end{theorem}

\begin{proof}
    The analysis depends on whether $P$ is live or not.
    \begin{itemize}
        \item
            If $P$ is not live, then, by \Cref{l:nLivenAction}, it does not contain any actions or prefixes.
            Any recursive loops in $P$ are thus of the form `$\mu X_1(); \ldots; \mu X_n(); \0$': contractiveness requires recursive calls to be prefixed by inputs/branches or bound to parallel outputs/selections, of which there are none.
            Hence, we can use structural congruence to rewrite each recursive loop in $P$ to $\0$ by unfolding, yielding $P' \equiv P$.
            The remaining derivation of $P'$ only contains applications of \scc{Empty}, \scc{Mix}, $\bullet$, or \scc{Cycle} on closed endpoints.
            It follows easily that $P \equiv P' \equiv \0$.

        \item
            If $P$ is live, by \Cref{t:dlFree} there is $Q$ s.t.\ $P \redd Q$.
            Moreover, $P$ does not have free names, for otherwise it would not be typable under empty context.
            Because commuting conversions apply only to free names, this means $P \redd_\beta Q$.
            \qedhere
    \end{itemize}
\end{proof}

\subsection{Explicit Closing and Replicated Servers}
\label{ss:exrules}

\begin{figure}[t]
    \begin{mdframed}
        \mbox{}\hfill%
        \begin{wfit}
            \begin{prooftree}
                \infAx{
                    $x[] \vdash \Omega; x{:}\1^\pri$
                }{$\1$}
            \end{prooftree}
        \end{wfit}%
        \hfill%
        \begin{wfit}
            \begin{prooftree}
                \infAss{
                    $P \vdash \Omega; \Gamma$
                }
                \infAss{
                    $\pri < \pr(\Gamma)$
                }
                \infBin{
                    $x(); P \vdash \Omega; \Gamma, x{:}\bot^\pri$
                }{$\bot$}
            \end{prooftree}
        \end{wfit}%
        \hfill\mbox{}

        \smallskip
        \mbox{}\hfill%
        \begin{wfit}
            \begin{prooftree}
                \infAx{
                    ${?}x[y] \vdash \Omega; x{:}{?}^\pri A, y{:}\ol{A}$
                }{${?}$}
            \end{prooftree}
        \end{wfit}%
        \hfill%
        \begin{wfit}
            \begin{prooftree}
                \infAss{
                    $P \vdash \Omega; {?}\Gamma, y{:}A$
                }
                \infAss{
                    $\pri < \pr({?}\Gamma)$
                }
                \infBin{
                    ${!}x(y); P \vdash \Omega; {?}\Gamma, x{:}{!}^\pri A$
                }{${!}$}
            \end{prooftree}
        \end{wfit}%
        \hfill\mbox{}

        \smallskip
        \mbox{}\hfill%
        \begin{wfit}
            \begin{prooftree}
                \infAss{
                    $P \vdash \Omega; \Gamma$
                }
                \infUn{
                    $P \vdash \Omega; \Gamma, x{:}{?}^\pri A$
                }{\scc{W}}
            \end{prooftree}
        \end{wfit}%
        \hfill%
        \begin{wfit}
            \begin{prooftree}
                \infAss{
                    $P \vdash \Omega; \Gamma, x{:}{?}^\pri A, x'{:}{?}^\kappa A$
                }
                \infAss{
                    $\pi = \min(\pri,\kappa)$
                }
                \infBin{
                    $P \subst{x/x'} \vdash \Omega; \Gamma, x{:}{?}^\pi A$
                }{\scc{C}}
            \end{prooftree}
        \end{wfit}%
        \hfill%
        \begin{wfit}
            \begin{prooftree}
                \infAss{
                    $P \vdash \Omega; \Gamma, y{:}A$
                }
                \infUn{
                    ${?}\ol{x}[y] \cdot P \vdash \Omega; \Gamma, x{:}{?}^\pri A$
                }{${?}^\star$}
            \end{prooftree}
        \end{wfit}%
        \hfill\mbox{}


        \medskip
        \hbox to \textwidth{\leaders\hbox to 3pt{\hss . \hss}\hfil}

        \vspace{-1.5em}
        \begin{align*}
            & \brred{\1\bot}
            &
            &
            &
            \nu{x y}(x[] \| y(); P)
            & \redd P
            \\
            & \brred{{?}{!}}
            &
            &
            &
            \nu{x y}({?}x[a] \| {!}y(v); P \| Q)
            & \redd P \subst{a/v} \| \nu{x y}({!}y(v); P \| Q)
            \\
            & \krred{\bot}
            &
            x \notin \tilde{v},\tilde{w} \implies
            &
            &
            \nu{\tilde{v}\tilde{w}}(x(); P \| Q)
            & \redd x(); \nu{\tilde{v}\tilde{w}}(P \| Q)
            \\
            & \krred{{!}}
            &
            x \notin \tilde{v},\tilde{w} \implies
            &
            &
            \nu{\tilde{v}\tilde{w}}({!}x(y); P \| Q)
            & \redd {!}x(y); \nu{\tilde{v}\tilde{w}}(P \| Q)
        \end{align*}
    \end{mdframed}

    \caption{Typing rules for explicit closing and replicated servers.}
    \label{f:exrules}
\end{figure}

As already mentioned, our presentation of APCP does not include explicit closing and replicated servers.
We briefly discuss what APCP would look like if we were to include these constructs.

We achieve explicit closing by adding empty outputs `$x[]$' and empty inputs `$x(); P$' to the syntax of \Cref{f:procdef} (top).
We also add the synchronization `$\brred{\1\bot}$' and the commuting conversion `$\krred{\bot}$' in \Cref{f:exrules} (bottom).
At the level of types, we replace the conflated type `$\bullet$' with `$\1^\pri$' and `$\bot^\pri$', associated to empty outputs and empty inputs, respectively.
Note that we do need priority annotations on types for closed endpoints now, because the empty input is blocking and thus requires priority checks.
In the type system of \Cref{f:apcpInf} (top), we replace rule `$\bullet$' with the rules `$\1$' and `$\bot$' in \Cref{f:exrules} (top).

For replicated servers, we add client requests `${?}x[y]$' and servers `${!}x(y); P$', typed `${?}^\pri A$' and `${!}^\pri A$', respectively.
We include syntactic sugar for binding client requests to continuations as in \Cref{n:sugar}: `${?}\ol{x}[y] \cdot P := \nu{y a}({?}x[a] \| P)$'.
New reduction rules are in \Cref{f:exrules} (bottom): synchronization rule `$\brred{{?}{!}}$', connecting a client and a server and spawns a copy of the server, and commuting conversion `$\krred{{!}}$'.
Also, we add a structural congruence axiom to clean up unused servers: $\nu{x z}({!}x(y); P) \equiv \0$.
In the type system, we add rules `${?}$', `${!}$', `\scc{W}' and `\scc{C}' in \Cref{f:exrules} (top); the former two are for typing client requests and servers, respectively, and the latter two are for connecting to a server without requests and for multiple requests, respectively.
In rule `${!}$', notation `${?}\Gamma$' means that every type in $\Gamma$ is of the form `${?}^\pri A$'.
\Cref{f:exrules} (top) also includes an admissible rule `${?}^\star$' which types the syntactic sugar for bound client requests.

\section{Examples}
\label{s:examples}

Up to here, we have presented our process language and its type system, and we have discussed the influence of asynchrony and recursion in their design and properties ensured by typing. We now present examples to further illustrate the design and expressiveness of APCP.

\subsection{Milner's Typed Cyclic Scheduler}
\label{ss:milnertyped}

\begin{figure}[t]
    \begin{mdframed}
        \begin{prooftree}
            \infAx{
                $\vdash X{:}3; a_1{:}X, c_n{:}X, d_1{:}X$
            }{\scc{Var}}
            \infUn{
                $\vdash X{:}3; a_1{:}X, c_n{:}\&^{\rho_n} \{\sff{next}: X\}, d_1{:}X$
            }{$\&$}
            \infUn{
                $\vdash X{:}3; a_1{:}X, c_n{:}\&^{\pi_n} \{\sff{start}: \&^{\rho_n} \{\sff{next}: X\}\}, d_1{:}X$
            }{$\&$}
            \infUn{
                $\vdash X{:}3; a_1{:}X, c_n{:}\&^{\pi_n} \{\sff{start}: \&^{\rho_n} \{\sff{next}: X\}\}, d_1{:}{\oplus}^{\rho_1} \{\sff{next}: X\}$
            }{$\oplus^\star$}
            \infUn{
                $\vdash X{:}3; a_1{:}\&^{\kappa_1} \{\sff{ack}: X\}, c_n{:}\&^{\pi_n} \{\sff{start}: \&^{\rho_n} \{\sff{next}: X\}\}, d_1{:}{\oplus}^{\rho_1} \{\sff{next}: X\}$
            }{$\&$}
            \infUn{
                $\vdash X{:}3; a_1{:}{\oplus}^{\pri_1} \{\sff{start}: \&^{\kappa_1} \{\sff{ack}: X\}\}, c_n{:}\&^{\pi_n} \{\sff{start}: \&^{\rho_n} \{\sff{next}: X\}\}, d_1{:}{\oplus}^{\rho_1} \{\sff{next}: X\}$
            }{$\oplus^\star$}
            \infUn{
                $\vdash X{:}3;~ a_1{:}{\oplus}^{\pri_1} \{\sff{start}: \&^{\kappa_1} \{\sff{ack}: X\}\},
                \begin{array}[t]{l}
                    c_n{:}\&^{\pi_n} \{\sff{start}: \&^{\rho_n} \{\sff{next}: X\}\},
                    \\
                    d_1{:}{\oplus}^{\pi_1} \{\sff{start}: {\oplus}^{\rho_1} \{\sff{next}: X\}\}
                \end{array}$
            }{$\oplus^\star$}
            \infUn{
                $\vdash \emptyset;~ a_1{:}\mu X.{\oplus}^{\pri_1} \{\sff{start}: \&^{\kappa_1} \{\sff{ack}: X\}\},
                \begin{array}[t]{l}
                    c_n{:}\mu X.\&^{\pi_n} \{\sff{start}: \&^{\rho_n} \{\sff{next}: X\}\},
                    \\
                    d_1{:}\mu X.\&^{\pi_1} \{\sff{start}: {\oplus}^{\rho_1} \{\sff{next}: X\}\}
                \end{array}$
            }{\scc{Rec}}
        \end{prooftree}
    \end{mdframed}

    \caption{Typing derivation of the leader scheduler $A_1$ of Milner's cyclic scheduler (processes omitted).}
    \label{f:leader}
\end{figure}

To consider a process that goes beyond the scope of PCP, here we show that our specification of Milner's cyclic scheduler from \Cref{s:milner} is typable in APCP, and thus deadlock free (cf.\ \Cref{t:closedDLFree}).
Let us recall the process definitions of the leader and followers, omitting braces `$\{\ldots\}$' for branches with one option:
\begin{align*}
    A_1
    &:= \mu X(a_1,c_n,d_1); d_1 \puts \sff{start} \cdot a_1 \puts \sff{start} \cdot a_1 \gets \sff{ack}; d_1 \puts \sff{next} \cdot c_n \gets \sff{start}; c_n \gets \sff{next}; X\call{a_1,c_n,d_1}
    \\
    A_{i+1}
    &:= \mu X(a_{i+1},c_i,d_{i+1}); \begin{array}[t]{l}
        c_i \gets \sff{start}; a_{i+1} \puts \sff{start} \cdot d_{i+1} \puts \sff{start} \cdot a_{i+1} \gets \sff{ack};
        \\
        c_i \gets \sff{next}; d_{i+1} \puts \sff{next} \cdot X\call{a_{i+1},c_i,d_{i+1}}
    \end{array}
    \qquad \forall 1 \leq i < n
\end{align*}
Each process $A_{i+1}$ for $0 \leq i < n$---thus including the leader---is typable as follows, assuming $c_i$ is $c_n$ for $i=0$ (see \figref{f:leader} for the derivation of $A_1$, omitting processes from judgments):
\begin{align*}
    A_{i+1}
    &\vdash \emptyset;~ a_{i+1}{:}\mu X.{\oplus}^{\pri_{i+1}} \{\sff{start}: \&^{\kappa_{i+1}} \{\sff{ack}: X\}\},
    \begin{array}[t]{l}
        c_i{:}\mu X.\&^{\pi_i} \{\sff{start}: \&^{\rho_i} \{\sff{next}: X\}\},
        \\
        d_{i+1}{:}\mu X.{\oplus}^{\pi_{i+1}} \{\sff{start}: {\oplus}^{\rho_{i+1}} \{\sff{next}: X\}\}
    \end{array}
\end{align*}
Note how, for each $1 \leq i \leq n$, the types for $c_i$ and $d_i$ are duals.

To verify these typing derivations, we need to assign values to the priorities $\pri_i,\kappa_i,\pi_i,\rho_i$ for each $1 \leq i \leq n$ that satisfy the necessary requirements.
From the derivation of $A_1$ we require $\kappa_1 < \rho_1, \pi_n$.
For each $1 \leq i < n$, from the derivation of $A_{i+1}$ we require $\rho_i < \rho_{i+1}$ and $\kappa_{i+1} < \rho_i, \rho_{i+1}$ and $\pi_i < \pri_{i+1}, \pi_{i+1}$.
We can easily satisfy these requirements by assigning $\pri_i := \kappa_i := \pi_i := i$ and $\rho_i := i+2$ for each $1 \leq i \leq n$.

Assuming that $P_i \vdash \emptyset; a_i{:}\mu X.\&^{\pri_i} \{\sff{start}: {\oplus}^{\kappa_i} \{\sff{ack}: X\}\}$ for each $1 \leq i \leq n$, we have $Sched_n \vdash \emptyset; \emptyset$.
Hence, it follows from \Cref{t:closedDLFree} that $Sched_n$ is deadlock free for each $n \geq 1$.

\subsection{Comparison to Padovani's Type System for Deadlock Freedom}
\label{ss:padovani}

Padovani's type system for deadlock freedom~\cite{conf/csl/Padovani14} simplifies a type system by Kobayashi~\cite{conf/concur/Kobayashi06}; both these works do not consider session types.
Just as for Dardha and Gay's PCP~\cite{conf/fossacs/DardhaG18}, the priority annotations of APCP are based on similar annotations in Padovani's and Kobayashi's type systems.
Here, we compare APCP to these type systems by discussing some of the examples in Padovani's work.

\paragraph{Ring of Processes}

To illustrate APCP's flexible support for recursion, we consider Padovani's ever-growing ring of processes~\cite[Ex.~3.8]{conf/csl/Padovani14}.
For the ring to continuously loop, Padovani uses self-replicating processes.
Although this exact method is not possible in APCP, we can use recursion instead:
\[
    Ring_x^y := \mu X(x, y); \nu{a a'}(x(z); \nu{b b'}(X\call{z,b} \| X\call{b',a}) \| \ol{y}[c] \cdot a' \fwd c) \vdash \emptyset; x{:}\mu X.X \parr^\pri \bullet, y{:}\mu X.X \tensor^\kappa \bullet
\]
Each iteration, this process receives a fresh endpoint from its left neighbor and sends another fresh endpoint to its right neighbor.
It then spawns two copies of itself, connected to each other on a fresh channel, and one connected to the left neighbor and the other to the right neighbor.
There are no priority requirements, so we can let $\pri = \kappa$.
We can then connect the initial copy of $Ring$ to itself, forming a deadlock free ring of processes that doubles in size at every iteration (cf.\ \Cref{t:closedDLFree}):
\[
    \begin{array}{c}
        \nu{x y}Ring_x^y
        \redd^3 \nu{x_1 y_1}\nu{x_2 y_2}(Ring_{x_1}^{y_2} \| Ring_{x_2}^{y_1})
        \\[4pt]
        \redd^6 \nu{x_1 y_1}\nu{x_2 y_2}\nu{x_3 y_3}\nu{x_4 y_4}(Ring_{x_1}^{y_2} \| Ring_{x_2}^{y_3} \| Ring_{x_3}^{y_4} \| Ring_{x_4}^{y_1})
        \redd^{12} \ldots
    \end{array}
\]

\paragraph{Blocking versus Non-blocking}

Padovani discusses the significance of blocking inputs versus non-blocking outputs~\cite[Exs.~2.2 \&~3.6]{conf/csl/Padovani14}.
Although we can express Padovani's example in APCP with minor modifications, we can do so more directly by including replication as in \Cref{ss:exrules}.
Consider the following processes, which are identical up to the order of input and output:
\[
    {Node}_A := {!}c_A(c); c(x); c(y); \ol{x}[a] \cdot y(z); \0
    \hspace{3em}
    {Node}_B := {!}c_B(c); c(x); c(y); y(z); \ol{x}[a] \cdot \0
\]
We consider several configurations of nodes, using the syntactic sugar $\ol{x}\<y\> \cdot P := \ol{x}[y'] \cdot (y \fwd y' \| P)$:
\begin{align*}
    L_1(X)
    &:= \nu{c_A c'_A}\nu{c_B c'_B}({Node}_A \| {Node}_B \| {?}\ol{c'_X}[c] \cdot \nu{e e'}(\ol{c}\<e\> \cdot \ol{c}\<e'\> \cdot \0))
    \\
    L_2(X,Y)
    &:= \nu{c_A c'_A}\nu{c_B c'_B}({Node}_A \| {Node}_B \| \nu{e e'}\nu{f f'}\left(\begin{array}{l}
            {?}\ol{c'_X}[c] \cdot \ol{c}\<e\> \cdot \ol{c}\<f\> \cdot \0
            \\
            \|\, {?}\ol{c'_Y}[c'] \cdot \ol{c'}\<f'\> \cdot \ol{c'}\<e'\> \cdot \0
    \end{array}\right))
    \\
    L_3(X,Y,Z)
    &:= \nu{c_A c'_A}\nu{c_B c'_B}({Node}_A \| {Node}_B \| \nu{e e'}\nu{f f'}\nu{g g'}\left(\begin{array}{l}
            {?}\ol{c'_X}[c] \cdot \ol{c}\<e\> \cdot \ol{c}\<f\> \cdot \0
            \\
            \|\, {?}\ol{c'_Y}[c'] \cdot \ol{c'}\<g\> \cdot \ol{c'}\<e'\> \cdot \0
            \\
            \|\, {?}\ol{c'_Z}[c''] \cdot \ol{c''}\<f'\> \cdot \ol{c''}\<g'\> \cdot \0
    \end{array}\right))
\end{align*}
where $X,Y,Z \in \{A,B\}$.

To illustrate the significance of APCP's asynchrony, let us consider how $L_2(A,A)$ reduces:
\[
    L_2(A, A) \redd^6 \nu{e e'}\nu{f f'}(\ol{e}[a] \cdot f(z); \0 \| \ol{f'}[a'] \cdot e'(z'); \0) \redd \nu{f f'}(f(z); \0 \| \ol{f'}[a'] \cdot \0) \redd \0.
\]
The synchronization on $e$ and $e'$ is possible because the output on $f'$ is non-blocking.
It is also possible for $f$ and $f'$ to synchronize first, because the output on $e$ is also non-blocking.
In contrast, the reduction of $L_2(B,B)$ illustrates the blocking behavior of inputs:
\[
    L_2(B,B) \redd^6 \nu{e e'}\nu{f f'}(f(z); \ol{e}[a] \cdot \0 \| e'(z'); \ol{f'}[a'] \cdot \0) \nredd.
\]
This results in deadlock, for each node awaits a message, blocking their output to the other node.

Let us show how APCP detects (freedom of) deadlocks in each of these configurations by considering priority requirements.
For $X \in \{A,B\}$, we have ${Node}_X \vdash \emptyset; c_X{:}{!}^\pri ((\bullet \tensor^{\kappa_X} \bullet) \parr^{\pi_X} (\bullet \parr^{\rho_X} \bullet) \parr^{\psi_X} \bullet)$, requiring $\rho_B < \kappa_B$ and $\pi_X,\psi_X < \kappa_X,\rho_X$.
In each configuration, the input endpoint of one node is connected to the output endpoint of another.
Duality thus requires that $\kappa_W = \rho_{W'}$ for $W,W' \in \{X,Y,Z\}$.
Hence, in any configuration, if the input endpoint of a $Node_B$ is connected the output endpoint of another $Node_B$, we require $\kappa_B = \rho_B$, violating the requirement that $\rho_B < \kappa_B$.
From this we can conclude that the above configurations are deadlock free if and only if at least one of $X,Y,Z$ is $A$, and at most one of them is $B$.
This verifies that $L_2(A,A)$ contains no deadlock, while $L_2(B,B)$ does.

Note that in PCP the conditions for deadlock freedom are much stricter, as PCP's blocking outputs additionally require that $\kappa_A < \rho_A$.
Hence, we also cannot connect the input of a $Node_A$ to the output of another $Node_A$.
This means that $L_2(A,B)$ and $L_2(B,A)$ are the only deadlock free configurations in PCP.

\section{Related Work \& Conclusion}
\label{s:concl}

We have already discussed several related works throughout the paper~\cite{conf/fossacs/DardhaG18,conf/csl/DeYoungCPT12,conf/concur/Kobayashi06,conf/csl/Padovani14}.
The work of Kobayashi and Laneve~\cite{DBLP:journals/iandc/0001L17} is related to APCP in that it addresses deadlock freedom for \emph{unbounded} process networks.
Another related approach is Toninho and Yoshida's~\cite{journal/toplas/ToninhoY18}, which addresses deadlock freedom for cyclic process networks by generating global types from binary types.
The work by Balzer \etal~\cite{conf/icfp/BalzerP17,conf/esop/BalzerTP19} is also worth mentioning: it guarantees deadlock freedom for processes with shared, mutable resources by means of manifest sharing, i.e.\ explicitly acquiring and releasing access to resources.
Finally, Pruiksma and Pfenning's session type system derived from adjoint logic~\cite{conf/places/PruiksmaP19,journal/jlamp/PruiksmaP20} treats asynchronous, non-blocking actions via axiomatic typing rules, similarly as we do (cf.\ axioms `$\tensor$' and `$\oplus$' in \Cref{f:apcpInf}); we leave a precise comparison with their approach for future work.

In this paper, we have presented APCP, a type system for deadlock freedom of cyclic process networks with asynchronous communication and recursion.
We have shown that, when compared to (the synchronous) PCP~\cite{conf/fossacs/DardhaG18}, asynchrony in APCP significantly simplifies the management of priorities required to detect cyclic dependencies (cf.\ the discussion at the end of \Cref{ss:padovani}).
We illustrated the expressivity of APCP using multiple examples, and concluded that it is comparable in expressivity to similar type systems not based on session types or logic, in particular the one by Padovani~\cite{conf/csl/Padovani14}.
More in-depth comparisons with this and the related type systems cited above would be much desirable.
Finally, in ongoing work we are applying APCP to the analysis of multiparty protocols implemented as processes~\cite{report/vdHeuvelP21}.

\paragraph{Acknowledgements}
We are grateful to the anonymous reviewers for their careful reading of our paper and their useful feedback.
We also thank Ornela Dardha for clarifying the typing rules of PCP to us.

\phantomsection
\addcontentsline{toc}{section}{References}
\bibliographystyle{eptcs}
\bibliography{refs}

\end{document}